\documentclass[conference]{IEEEtran}
\usepackage{amsmath}
\usepackage{amssymb}
\usepackage{amsthm}
\usepackage{isomath}
\usepackage{caption}
\usepackage[dvipsnames]{xcolor}
\usepackage{enumitem}
\usepackage{hyperref}
\usepackage{cite}
\usepackage{graphicx}

\DeclareMathOperator{\E}{\mathbb{E}}
\DeclareMathOperator{\Var}{\mathrm{Var}}
\newtheorem{lemma}{Lemma}

\newtheorem{proposition}{Proposition}

\hypersetup{
	colorlinks,
	linkcolor={red!50!black},
	citecolor={blue!50!black},
	urlcolor={blue!80!black}
}

\newcommand{\matr}[1]{\boldsymbol{#1}}

\newcommand{\bigzero}{\mathbf{0}}

\author{\IEEEauthorblockN{Olle Abrahamsson, Danyo Danev and Erik G. Larsson}
	\IEEEauthorblockA{Dept. of Electrical Engineering (ISY), \\Link{\"o}ping University, 58183 Link{\"o}ping, Sweden\\
Email: \{olle.abrahamsson, danyo.danev, erik.g.larsson\}@liu.se}}
\title{Opinion Dynamics with Random Actions and a Stubborn Agent}

\begin{document}
	\maketitle
	\begin{abstract}
		We study opinion dynamics in a social network with stubborn agents who influence their neighbors but who themselves always stick to their initial opinion. We consider first the well-known DeGroot model. While it is known in the literature that this model can lead to consensus even in the presence of a stubborn agent, we show that the same result holds under weaker assumptions than has been previously reported. We then consider a recent extension of the DeGroot model in which the opinion of each agent is a random Bernoulli distributed variable, and by leveraging on the first result we establish that this model also leads to consensus, in the sense of convergence in probability, in the presence of a stubborn agent. Moreover, all agents' opinions converge to that of the stubborn agent.
	\end{abstract}
	\section{Introduction}
	The study of opinion dynamics in social networks goes back several decades; for a review, see e.g. \cite{proskurnikov2017,proskurnikov2018}. One of the most well-known models is the DeGroot model \cite{degroot} which has been studied extensively (for a literature survey see for instance \cite[Section 3]{proskurnikov2017} and \cite[Section 3]{proskurnikov2018}). In this model an agent's opinion is represented by a continuous real variable, which at each time step \(n\in\{1,2,\dots\}\) is updated to a linear combination of the opinions of itself and its neighbors,
	\begin{equation}\label{eq:degroot_update}
	\matr{x}[n+1] = \matr{T}\matr{x}[n],
	\end{equation}
	where \(\matr{x}[n]\) represents the agents' opinions at time \(n\) and \(\matr{T}\) is a matrix that encodes the trust between agents (this is explained in detail in Section \ref{sec:DeGroot}).
	
	A particular case in opinion dynamics is where one or more agents are stubborn (agents whose opinions remain unchanged independent of the others' opinions). This scenario was first introduced by Mobilia in 2003 \cite{mobilia2003} who established convergence rates towards consensus under the so-called voter model \cite{sudbury73} with a single stubborn agent. The voter model was again considered in \cite{yildiz} where the optimal placement of stubborn agents for maximal influence on the long-term expected opinions was investigated, among other properties. In \cite{ghaderi} the authors considered a model in which agents can have a continuous degree of stubbornness, and gave bounds on the rate of convergence to a consensus of opinions. A more recent study \cite{wai} showed that the influence of stubborn agents under the DeGroot model can, under suitable conditions, be used to recover the topology of the network. Specifically the authors derived equations for the expected opinions of the ordinary (non-stubborn) agents that depend on the topology, and then showed how a regression problem could be formulated which estimated matrices with information about the topology by observing opinions that fit the equations.
	
	Another class of models incorporates randomness, for example in terms of random interactions \cite{acemoglu2010,mukhopadhyay2016} or as in \cite{hunter2018}, where at each time \(n\) a randomly selected agent communicates a random opinion to its neighbors. The latter model also features the interesting novelty that an agent may grow increasingly stubborn over time. A recent extension of the DeGroot model which incorporates randomness was given in \cite{scaglione}. Under this setting, at every time step \(n\) each agent \(k\) chooses a Bernoulli distributed random action \(A_k[n]~\sim~\mathrm{Bernoulli}(X_k[n])\), and the corresponding update rule is
	\begin{equation}\label{eq:RA_update}
	\matr{X}[n+1] = (1-\alpha)\matr{X}[n] + \alpha\matr{T}\matr{A}[n],
	\end{equation}
	as described further in Section \ref{sec:RA}.\footnote{We use uppercase letters for random variables, e.g., \(\matr{X}[n]\). They are distinguishable from matrices (which are deterministic), e.g., \(\matr{T},\matr{Q}\), since the matrices are not time dependent.} In this model, which we will refer to as the Random Actions model (RA model for short), the probabilities of the actions, rather than the actions/opinions themselves, are updated as a weighted average over the neighbors' actions.
	\section{Contributions}
	We extend the RA model by the introduction of a stubborn agent and establish that the opinion dynamics converges in probability to a consensus even under this restriction, and furthermore that all agents adopt the stubborn agent's initial opinion. While this result is intuitively expected, the proof entails some non-trivial mathematical techniques.
	
	As a stepping stone towards the analysis of the aforementioned model we first consider the DeGroot model with a stubborn agent as described in \cite{wai} and show that the convergence results from that paper can be obtained with weakened assumptions on the model. Specifically, instead of assuming that every ordinary agent has a non-zero trust in the stubborn agent, it suffices to assume that at least one ordinary agent has such a trust. We then use this conclusion in proving the claims regarding the extended RA model. 
	\section{Models and Definitions}\label{sec:model}
	In both models described in this section, we will consider a directed, weighted, single-component network with \(K\) nodes, where the nodes are interpreted as agents. Before giving the details of the models, let us at this point remind the reader of some definitions. A \textit{sub-(row)-stochastic} matrix is a square, non-negative matrix such that the row sums are less than or equal to \(1\). The word ``row'' will be omitted and implied from hereon. There are two special cases of these matrices: A \textit{stochastic} matrix is a sub-stochastic matrix where all rows sum to \(1\), and a \textit{strictly sub-stochastic} matrix is a sub-stochastic matrix whose row sums are all strictly less than \(1\).
	\subsection{The DeGroot Model with a Stubborn Agent}\label{sec:DeGroot}
	In the DeGroot model \cite{degroot}, at every time step \(n\in\{0,1,\dots\}\), each agent \(k\in\{1,2,\dots,K\}\) observes the opinions of its neighbors, and updates its opinion to a linear combination of its own opinion and those of its neighbors. The update rule is given by Equation \eqref{eq:degroot_update} where \(\matr{x}[0]\in\mathbb{R}^K\) is a column vector representing the initial opinions of the \(K\) agents and \(\matr{T}\) is a \(K\times K\) stochastic matrix representing the trusts between agents. If \(\matr{T}^n\) converges to a limit \(\matr{T}^{\infty}\) as \(n\to\infty\), then consensus is reached and is given by
	\begin{equation}
	\lim_{n \to \infty}\matr{x}[n] = \lim_{n \to \infty}\matr{T}^n\matr{x}[0] = \matr{T}^{\infty}\matr{x}[0].
	\end{equation}
	If the agents are viewed as nodes in a network, then \(\matr{T}\) is interpreted as an adjacency matrix with elements \(t_{ij}\), and we use the convention that \(t_{ij}>0\) represents an edge from \(j\) to \(i\) whose weight is equal to the trust that \(i\) puts in \(j\).
	
	A special case is when one agent is stubborn, that is, an agent who never updates its opinion, corresponding to a node with only outgoing edges. Let the agents' opinions be partitioned into two sets of opinions, \(\matr{x}_1[n]\) and \(\matr{y}[n] = (x_{2}[n], x_{3}[n], \dots,x_K[n])^T\), held respectively by a stubborn agent and \(K-1\) ordinary agents. Then we write
	\begin{equation}\label{eq:partition}
	\matr{x}[n] = 
	\begin{pmatrix}
	x_1[n] \\ \matr{y}[n]
	\end{pmatrix}.
	\end{equation}
	In this case the trust matrix \(\matr{T}\) has the structure
	\begin{equation}\label{eq:degroot_trust}
	\matr{T} = \begin{bmatrix}
	1 & \bigzero\\
	\matr{r} & \matr{Q}
	\end{bmatrix},
	\end{equation}
	where the scalar \(1\) represents the stubborn agent, the vector \(\matr{r}\) with dimensions \((K-1)\times 1\) represents the links from stubborn to ordinary agent, and the matrix \(\matr{Q}\) represents the edges between ordinary agents. We will assume that all ordinary agents are strongly connected, so that \(\matr{Q}\) is irreducible.
	
	\subsection{The RA Model with a Stubborn Agent}\label{sec:RA}
	In the RA model \cite{scaglione}, at every time step each agent \(k\in\{1,2,\dots,K\}\) chooses one of two actions, \(0\) or \(1\), and these actions are generated by a Bernoulli random variable \(A_k[n]\) with probability \(X_k[n]\). The update of these probabilities is governed by Equation~\eqref{eq:RA_update}, where \(\alpha\in(0,1)\), \(\matr{T}\) is a trust matrix (as defined in Section \ref{sec:DeGroot}), and \(\matr{A}[n] = (A_1[n], A_2[n], \dots, A_K[n])^T\) are the actions with corresponding probabilities \(\matr{X}[n] = (X_1[n], X_2[n], \dots, X_K[n])^T \in[0,1]^K\), which themselves are stochastic for \(n>0\).
	
	In the case with a stubborn agent we can assume w.l.o.g. that this agent always chooses action \(0\) with probability \(1\). Analogously to equations \eqref{eq:partition} and \eqref{eq:degroot_trust}, we have
	\begin{equation}\label{eq:RA_stubborn}
	\matr{X}[n] = \begin{pmatrix}
	0 \\ \matr{Y}[n] \end{pmatrix}, \ 
	\matr{T} = \begin{bmatrix}
	1 & \bigzero\\
	\matr{r} & \matr{Q}
	\end{bmatrix}, \
	\matr{A}[n] = \begin{pmatrix}
	A_1[n] \\ \matr{B}[n]
	\end{pmatrix}
	\end{equation}
	where \(\matr{r}\) has dimension \((K-1)\times 1\) and \(\matr{B}[n] = (A_2[n], A_3[n], \dots, A_K[n])^T\).
	Then \(A_1[n] = 0\) with probability \(1\) for all \(n\geq 0\) and the other agents will update as in the original RA model. Again we assume that \(\matr{Q}\) is irreducible.
	\section{Results}
	The first proposition establishes the conditions for convergence of the model in Section \ref{sec:DeGroot} with trust matrix \(\matr{T}\) as defined in \eqref{eq:degroot_trust}.
		\begin{proposition}\label{Tconverge}
		If at least one ordinary agent puts a non-zero trust in the stubborn agent, that is, \(t_{i1}>0\) for some \(i>1\), then the limit \(\matr{T}^n, n\to\infty\), exists and has the structure
		\begin{equation*}
		\matr{T}^{\infty} =
		\begin{bmatrix}
		1 & \bigzero \\
		\matr{(I-Q)}^{-1}\matr{r} & \bigzero
		\end{bmatrix}.
		\end{equation*}
		\end{proposition}
	For the proof of Proposition \ref*{Tconverge} we need the following lemma.
	\begin{lemma}\label{lem:spectralradius}
		Let \(\matr{A}\) be an \(M\times M\) irreducible sub-stochastic matrix with at least one row sum being strictly less than \(1\), and let \(\rho(\matr{A})\) be the spectral radius of \(\matr{A}\). It holds that \(\rho(\matr{A})<1\).
	\end{lemma}
	The proofs of all lemmas in this article are given in the appendix. Note that for a strictly sub-stochastic matrix \(\matr{A}\), we can remove the assumption of irreducibility since it follows directly from Theorem 8.1.22 in \cite{HorJoh} that \(\rho(\matr{A})<1\).
	\begin{proof}[Proof of Proposition \ref{Tconverge}]
		The \(n\)th power of \(\matr{T}\) is
		\begin{equation}\label{eq: tn}
		\matr{T}^n =
		\begin{bmatrix}
		1 & \bigzero \\
		(\matr{I}+\matr{Q}+\matr{Q}^2+\dots +\matr{Q}^{n-1})\matr{r} & \matr{Q}^n
		\end{bmatrix},
		\end{equation}
		where \(\matr{Q}\) is sub-stochastic with at least one row having sum strictly less than \(1\). This is due to the assumption that \(t_{i1}>0\) for some \(i>1\), and since \(\matr{T}\) is stochastic, the \(i\)th row of \(\matr{Q}\) must have sum less than \(1\). Finally, since \(\matr{Q}\) is irreducible, Lemma \ref{lem:spectralradius} applies, and we have \(\rho(\matr{Q})<1\). By Theorem 5.6.12 in \cite{HorJoh}, this implies that
		\begin{equation}\label{Qconverge}
		\lim_{n \to \infty}\matr{Q}^n=\bigzero.
		\end{equation}
		Now, consider
		\begin{equation}\label{partial}
		(\matr{I}-\matr{Q})\sum_{k=0}^{n-1}\matr{Q}^k = \sum_{k=0}^{n-1}\left(\matr{Q}^k - \matr{Q}^{k+1}\right) = \matr{I} - \matr{Q}^{n}.
		\end{equation}
		By Equation \eqref{Qconverge}, the right hand side of \eqref{partial} tends to \(\matr{I}\) in the limit as \(n\to\infty\), and since \(\rho(\matr{Q})<1\) the matrix \(\matr{I}-\matr{Q}\) is invertible.\footnote{To see this, suppose \(\matr{I}-\matr{Q}\) is not invertible. Then there exists a non-zero vector \(\matr{v}\) such that \((\matr{I}-\matr{Q})\matr{v} = \bigzero\), or equivalently \(\matr{Qv} = \matr{v}\), which shows that \(1\) is an eigenvalue of \(\matr{Q}\). But this is impossible since \(\rho(\matr{Q})<1\).}   It follows that
		\begin{equation}
		\lim_{n \to \infty}\sum_{k=0}^{n-1}\matr{Q}^k=(\matr{I}-\matr{Q})^{-1}.
		\end{equation}
	\end{proof}
	With the previously discussed decomposition of \(\matr{x}[n]\) into stubborn and ordinary agents in \eqref{eq:partition}, the opinions of ordinary agents converge as \(n\to\infty\):
	\begin{equation}\label{consensus2}
	\lim_{n \to \infty}\ \matr{y}[n] = (\matr{I}-\matr{Q})^{-1}\matr{r}x_1[0].
	\end{equation}
	
	The second proposition concerns the RA model in Section \ref{sec:RA} with a stubborn agent.
		\begin{proposition}\label{RA-prop}
		The opinion dynamics of \eqref{eq:RA_update} under the restrictions imposed by \eqref{eq:RA_stubborn} leads to herding in the sense of convergence in probability, i.e., for every \(\varepsilon>0\),
		\begin{equation*}
		\lim_{n \to \infty}\mathbb{P}(X_k[n]<1-\varepsilon) = 0, \ \text{for all } k\in\{1,2,\dots,K\}.
		\end{equation*}
		\end{proposition}
The first part of the proof of Proposition \ref{RA-prop} treats the convergence of opinions towards a consensus in the subnetwork induced by the ordinary agents and follows partly the proof of Theorem 1 in \cite{scaglione}, but with some modifications due to the presence of the stubborn agent. The second part shows that the consensus opinion must be equal to that of the stubborn agent. In this part we deviate from \cite{scaglione} in that we show convergence in probability, as opposed to the claimed proof of almost sure convergence therein, which we have been unable to verify. A detailed discussion of the differences will be provided elsewhere. We need the following facts for the main proof.
	\begin{lemma}\label{lem:ms-implies-prob-conv}
		If \(\{W[n]\}_{n=0}^\infty\) is a sequence of random variables such that \(W[n]\in [0,1]\) for all \(n\geq 0\), and
		\[\lim_{n \to \infty}\E[W[n]^2(1-W[n])^2] = 0,\]
		then for all \(\varepsilon > 0\),
		\[\lim_{n\to\infty}\mathbb{P}(W[n]\leq\varepsilon \ \cup \ W[n]\geq 1-\varepsilon) = 1.\]
	\end{lemma}
	\begin{lemma}\label{lem:conv-prob}
		Consider the update rule in \eqref{eq:RA_update} with \(\matr{X}[n] = (X_1[n], X_2[n], \dots, X_K[n])^T\). Suppose agent \(i\) puts trust in agent \(j\) (so that \(t_{ij}>0\)). If \(X_j[n]\xrightarrow{P}0\) and
		\[\lim_{n \to \infty}\E[X_i[n]^2(1-X_i[n])^2] = 0,\] 
		then \(X_i[n]\xrightarrow{P}0\).
	\end{lemma}
	\begin{proof}[Proof of Proposition \ref{RA-prop}]
	 Let \(\matr{Y}[n]\), \(\matr{B}[n]\), \(\matr{r}\) and \(\matr{Q}\) be defined as in \eqref{eq:RA_stubborn}. Since the vector \(\matr{r}\) has at least one positive element, \(\matr{Q}\) is sub-stochastic with at least one row sum strictly less than one, so by Lemma \ref{lem:spectralradius} it has a largest eigenvalue \(\lambda\in(0,1)\) with corresponding left eigenvector \(\matr{\psi}\), \(\matr{\psi}^T\matr{Q} = \lambda\matr{\psi^T}\). Let \(S[n] = \matr{\psi}^T\matr{Y}[n]\). The proof will proceed as follows: First we show that \(S[n]\) is a strict super-martingale that converges in the limit as \(n\to\infty\) to a random variable \(S[\infty]\). Then we show that the conditional variance of the martingale difference sequence \(S[n]-S[n-1]\) converges to zero in the mean square sense. We conclude that all elements in \(\matr{Y[n]}\) converge in probability to the value of the stubborn agent, \(X_1[0] = 0\).
	 
	 We will now show that \(S[n]\) is a strict super-martingale w.r.t. \(\matr{Y}[n]\), that is, \(\E[S[n+1]\mid \matr{Y}[n]] < S[n]\). First, note that by the update rule in Equation \eqref{eq:RA_update},
	\begin{equation}\label{eq:ordinary_update}
	\matr{X}[n+1] = \begin{bmatrix}0 \\ \matr{Y}[n+1] \end{bmatrix} = (1-\alpha)\begin{bmatrix}0 \\ \matr{Y}[n] \end{bmatrix} + \alpha \matr{T}\begin{bmatrix}0 \\ \matr{B}[n] \end{bmatrix}.
	\end{equation}
 	Then we have
	\begin{equation}
	S[n+1] = \matr{\psi}^T\matr{Y}[n+1] = \matr{\psi}^T\big((1-\alpha)\matr{Y}[n] + \alpha\matr{Q}\matr{B}[n]\big),
	\end{equation} 
	and by taking expectations of both sides conditioned on \(\matr{Y}[n]\) we obtain
	\begin{equation}\label{eq:supmg-prop}
	\begin{aligned}
	\E[S[n+1]|\matr{Y}[n]] &= (1-\alpha)\matr{\psi}^T\matr{Y}[n] + \alpha\lambda\matr{\psi}^T\matr{Y}[n]\\
	&= (1-\alpha(1-\lambda))S[n] < S[n],
	\end{aligned}
	\end{equation}
	since \((1-\lambda) \in (0,1)\) and \(\alpha\in(0,1)\). Thus \(S[n]\) is a strict super-martingale, and since \(S[n] \geq 0\) for all \(n\) it follows from the Martingale Convergence Theorem \cite[Theorem 4.2.12]{Durrett} that
	\begin{equation}\label{eq:supmg-as-conv}
	S[n] \xrightarrow{a.s.} S[\infty],
	\end{equation}
	for some random variable \(S[\infty]\) as \(n\to\infty\).
	
	Consider now the martingale difference sequence \(\mathrm{\Delta} S[n] = S[n] - S[n-1]\) for \(n>1\). First note that the almost sure convergence of \(S[n]\) in Equation \eqref{eq:supmg-as-conv} implies
	\begin{equation}\label{eq:ds-as-conv}
	\mathrm{\Delta}S[n] \xrightarrow{a.s.} 0, \ n\to\infty.
	\end{equation}
	Furthermore, \(\matr{Q}\) is irreducible and non-negative, so by the Perron-Frobenius Theorem  \cite[Theorem 8.4.4]{HorJoh} all elements of \(\matr{\psi}\) are positive. Let \(\matr{\psi}\) be normalized so that \(\mathbf{1}^T\matr{\psi} = 1\), where \(\mathbf{1}^T = (1,1,\dots,1)\). Since \(Y_k[n]\in[0,1], \ k=1,2,\dots,K-1\), for all \(n\geq 0\) we then have \(S[n] \leq1\) and \(\lvert \mathrm{\Delta}S[n]\rvert = \lvert \matr{\psi}^T(\matr{Y}[n]-\matr{Y}[n-1])\rvert \leq 1\).
	Therefore, by the Dominated Convergence Theorem \cite[Theorem 1.5.8]{Durrett} together with the almost sure convergence in \eqref{eq:ds-as-conv}, \(\mathrm{\Delta}S[n]\) converges to \(0\) in \(m\)th mean, i.e.,
	\begin{equation}\label{m.s.con}
	\lim_{n \to \infty}\E[\lvert\mathrm{\Delta} S[n]\rvert^m] = 0, \ \text{for all } m\geq 1.
	\end{equation}
	We will now show that the variance of \(\mathrm{\Delta}S[n+1]\) conditioned on \(\matr{Y}[n]\) converges to zero in the mean square sense as \(n\to\infty\), and then conclude that the elements of \(\matr{Y}[n]\) converge in probability to all \(0\)s or all \(1\)s. We have:
	\begin{equation}\label{eq:ds-var-cond}
	\begin{aligned}
	&\Var(\mathrm{\Delta}S[n+1]\mid\matr{Y}[n]) \\
	&= \E\big[\big(\mathrm{\Delta}S[n+1] - \E[\mathrm{\Delta}S[n+1]\mid\matr{Y}[n]]\big)^2 \mid\matr{Y}[n]\big] \\
	&= \E\big[\big(\matr{\psi}^T(\matr{Y}[n+1] - \matr{Y}[n]) - \\
	& \ \ \matr{\psi}^T\E[\matr{Y}[n+1] - \matr{Y}[n]\mid\matr{Y}[n]]\big)^2 \mid\matr{Y}[n]\big] \\
	&= \E\big[\big(\matr{\psi}^T\matr{Y}[n+1] - \matr{\psi}^T\E[\matr{Y}[n+1]\mid\matr{Y}[n]]\big)^2 \mid\matr{Y}[n]\big] \\
	&= \E\big[\big(\matr{\psi}^T((1-\alpha)\matr{Y}[n] + \alpha \matr{Q}\matr{B}[n]) \\
	& \ \ - \matr{\psi}^T((1-\alpha)\matr{Y}[n] + \alpha \matr{Q}\matr{Y}[n])\big)^2 \mid\matr{Y}[n]\big] \\
	&= \E\big[\big(\alpha \matr{\psi}^T\matr{Q}(\matr{B}[n] - \matr{Y}[n])\big)^2 \mid\matr{Y}[n]\big] \\
	&= \alpha^2\lambda^2\matr{\psi}^T\E\big[(\matr{B}[n] - \matr{Y}[n])(\matr{B}[n] - \matr{Y}[n])^T \mid\matr{Y}[n]\big]\matr{\psi},
	\end{aligned}
	\end{equation}
	where in the last step we used that \(\matr{\psi}\) is a left eigenvector to \(\matr{Q}\) with eigenvalue \(\lambda\). The actions \(\matr{B}[n]\sim\textrm{Bernoulli}(\matr{Y}[n])\) are statistically independent conditioned on \(\matr{Y}[n]\), so only the diagonal elements of the covariance matrix \(\E\big[(\matr{B}[n] - \matr{Y}[n])\cdot(\matr{B}[n] - \matr{Y}[n])^T \mid\matr{Y}[n]\big]\) are non-zero. They can be expressed explicitly as
	\begin{equation}
	\begin{aligned}
	&\E\big[B_k^2[n]\mid\matr{Y}[n]\big] - \big(\E\big[B_k[n]\mid\matr{Y}[n]\big]\big)^2\\
	&=Y_k[n] - Y_k[n]^2\\
	&= Y_k[n](1-Y_k[n]), \ \text{for all } k=1,2,\dots,K-1.
	\end{aligned}
	\end{equation}
	Therefore,
	\begin{equation}\label{eq:ds-var-cond2}
	\begin{aligned}
	&\Var(\mathrm{\Delta}S[n+1]\mid\matr{Y}[n]) \\
	&= \alpha^2\lambda^2\sum_{k=1}^{K-1}\psi_k^2Y_k[n](1-Y_k[n]).
	\end{aligned}
	\end{equation}
	To see that the left hand side of \eqref{eq:ds-var-cond2} converges to zero in the mean square sense, consider its square:
	\begin{equation}\label{eq:square-var-ds}
	\begin{aligned}
	&(\Var(\mathrm{\Delta} S[n+1]\mid \matr{Y}[n]))^2 \\
	&= \big(\E[(\mathrm{\Delta} S[n+1])^2\mid \matr{Y}[n]] - (\E[\mathrm{\Delta}S[n+1]\mid\matr{Y}[n]])^2\big)^2 \\
	&= \big(\E[(\mathrm{\Delta} S[n+1])^2\mid \matr{Y}[n]]\big)^2 +\big(\E[\mathrm{\Delta}S[n+1]\mid\matr{Y}[n]]\big)^4 \\
	& \ \ \  - 2\E[(\mathrm{\Delta} S[n+1])^2\mid \matr{Y}[n]]\big(\E[\mathrm{\Delta}S[n+1]\mid\matr{Y}[n]]\big)^2\\
	&\leq \big(\E[(\mathrm{\Delta} S[n+1])^2\mid \matr{Y}[n]]\big)^2 +\big(\E[\mathrm{\Delta}S[n+1]\mid\matr{Y}[n]]\big)^4\\
	&\leq \E[(\mathrm{\Delta} S[n+1])^4\mid \matr{Y}[n]] +\E[(\mathrm{\Delta}S[n+1])^4\mid\matr{Y}[n]] \\
	&= 2\E[(\mathrm{\Delta}S[n+1])^4\mid\matr{Y}[n]],
	\end{aligned}
	\end{equation}
	where the first inequality holds since \((\mathrm{\Delta}S[n+1])^2\) is non-negative, and the second inequality is due to Jensen's inequality \cite[Theorem 1.6.2]{Durrett}. By taking expectations on both sides of \eqref{eq:square-var-ds} and using the result of convergence in \(m\)th mean in \eqref{m.s.con}, we obtain
	\begin{equation}
	\begin{aligned}
	& \lim_{n \to \infty}\E[(\Var(\mathrm{\Delta} S[n+1]\mid \matr{Y}[n]))^2] \\
	&\leq 2 \lim_{n \to \infty} \E\big[\E[(\mathrm{\Delta}S[n+1])^4\mid\matr{Y}[n]]\big] \\
	&= 2 \lim_{n \to \infty}\E[(\mathrm{\Delta}S[n+1])^4] = 0.
	\end{aligned}
	\end{equation}
	As already noted, all elements of \(\matr{\psi}\) are positive which, in view of Equation \eqref{eq:ds-var-cond2} together with the mean square convergence just proved, means that
	\begin{equation}
	\lim_{n \to \infty} \E[(Y_k[n](1-Y_k[n]))^2] = 0, \ \text{for all } k = 1,2,\dots,K-1.
	\end{equation}
	By Lemma \ref{lem:ms-implies-prob-conv} this implies that for all \(Y_k[n], k=1,2,\dots,K-1\) and for all \(\varepsilon>0\), we have
	\begin{equation}\label{eq:ordinary-conv1}
	\lim_{n\to\infty}\mathbb{P}(Y_k[n]<\varepsilon \ \cup \ Y_k[n]>1-\varepsilon) = 1.
	\end{equation}
	
   Let the set of ordinary agents be denoted by \(\mathcal{O}\), and define \(V_0\) as the subset of ordinary agents who put a trust in the stubborn agent, i.e., \(V_0 = \{i\in\mathcal{O}\mid t_{i1}>0\} \subseteq \mathcal{O}\); let \(V_1 \subseteq \mathcal{O}\setminus V_0\) denote the set of ordinary agents who put a trust in at least one of the agents in \(V_0\), and so on. Then by Lemma \ref{lem:conv-prob} together with \eqref{eq:ordinary-conv1} it follows that the elements in \(\{Y_k[n]\mid k\in V_0\}\) must converge in probability to \(0\). Consequently, the elements in \(\{Y_k[n]\mid k\in V_1\}\) must again converge to \(0\). Since \(\matr{Q}\) is irreducible there is some index \(P\) such that the union of the disjoint sets \(V_1,V_2,\dots,V_P\) makes up the set of ordinary agents, i.e.,
	\begin{equation}
	\bigcup_{p=1}^P V_p = \mathcal{O}.
	\end{equation}
	By continuing in this fashion it therefore follows that all elements in \(\{Y_k[n]\mid k\in\mathcal{O}\}\) must converge in probability to the value of the stubborn agent, \(X_1[n]=0\).	
\end{proof}
\appendix
\section{Proofs of Lemmas}
\begin{proof}[Proof of Lemma \ref{lem:spectralradius}]
	Let \(\mathbf{1} = (1,1,\dots,1)^T\) and for any \(m, \ 1\leq m\leq M,\) let \(r_m^{(n)} = [\matr{A}^n\mathbf{1}]_m\) be the \(m\)-th row sum of \(\matr{A}^n = \{a_{ij}^{(n)}\}\). Since \(\matr{A}\) is sub-stochastic we have that \(0\leq r_m^{(1)} \leq 1\) for all \(m\), and further that for any \(n\geq1\),
	\begin{equation}\label{eq:recursion}
	\begin{aligned}
		&r_m^{(n+1)} = \sum_{j=1}^{M}a_{mj}^{(n+1)} = \sum_{j=1}^{M}\left(\sum_{k=1}^{M}a_{mk}^{(n)}a_{kj}\right)\\
		&=\sum_{k=1}^{M}\left(a_{mk}^{(n)}\sum_{j=1}^{M}a_{kj}\right) = 
		\sum_{k=1}^{M}a_{mk}^{(n)}r_k^{(1)}.
	\end{aligned}
	\end{equation}
	Therefore
		\begin{equation}\label{eq:non-incr}
	\begin{aligned}
	&r_m^{(n+1)} \leq \sum_{k=1}^{M}a_{mk}^{(n)} = r_m^{(n)},
	\end{aligned}
	\end{equation}
	so the row sums are non-increasing with powers of \(\matr{A}\). By assumption at least one row sum is strictly less than \(1\), so w.l.o.g. we can assume that the rows of \(\matr{A}\)  are ordered such that this applies to the first row sum, i.e., \(r_1^{(1)} < 1\). By the irreducibility of \(\matr{A}\), for any \(m\) there is a positive integer \(l_m\) such that \(a_{m1}^{(l_m)}>0\) (since the induced network is strongly connected). In fact, if \(m\not= 1\) we have \(l_m<M\) (take the shortest path from node \(m\) to node \(1\)). By using \eqref{eq:recursion} we therefore obtain, for any row \(m\),
	\begin{equation}
	\begin{aligned}
	&r_m^{(l_m+1)} = \sum_{j=1}^{M}a_{mj}^{(l_m)}r_j^{(1)} = \sum_{j=2}^{M}a_{mj}^{(l_m)}r_j^{(1)} + a_{m1}^{(l_m)}r_1^{(1)} \\
	&\leq \sum_{j=2}^{M}a_{mj}^{(l_m)} + a_{m1}^{(l_m)}r_1^{(1)} < \sum_{j=1}^{M}a_{mj}^{(l_m)} = r_m^{(l_m)},
	\end{aligned}
	\end{equation}
	which together with \eqref{eq:non-incr} shows that every row sum of \(\matr{A}^n\) is strictly less than \(1\) for all \(n\geq M\). By Theorem 8.1.22 in \cite{HorJoh}, the spectral radius of a non-negative matrix is bounded from above by the maximum row sum.
	This means that \(\rho(\matr{A}^M) <1\), and since \(\rho(\matr{A}^M) = \rho(\matr{A})^M\), we therefore obtain \(\rho(\matr{A})<1\).
\end{proof}
\begin{proof}[Proof of Lemma \ref{lem:ms-implies-prob-conv}]
	Let \(\mu > 0\), and set \(\gamma = \mu\varepsilon^2(1-\varepsilon)^2\). We know that
	\begin{equation}
	\begin{aligned}
	&\lim_{n \to \infty}\E[W[n]^2 (1-W[n])^2]\\ 
	= &\lim_{n \to \infty}\int_0^1 w^2 (1-w)^2 f_{W[n]}(w)\mathrm{d}w = 0,
	\end{aligned}
	\end{equation}
	where \(f_{W[n]}(w)\) is the probability density function of \(W[n]\). Thus there exists \(N>0\) such that  \(A+B+C < \gamma\) for \(n\geq N\), where
	\begin{equation}
	\begin{aligned}
	&A=\int_0^\varepsilon w^2(1-w)^2 f_{W[n]}(w) \mathrm{d}w,\\
	&B=\int_\varepsilon^{1-\varepsilon} w^2 (1-w)^2 f_{W[n]}(w) \mathrm{d}w,\\
	&C=\int_{1-\varepsilon}^1 w^2(1-w)^2 f_{W[n]}(w) \mathrm{d}w.
	\end{aligned}
	\end{equation}
	But \(A>0\) and \(C>0\), so \(B< \gamma\) for all \(n\geq N\), and 
	\begin{equation}
	\gamma > B \geq \varepsilon^2 (1-\varepsilon)^2  \int_\varepsilon^{1-\varepsilon} f_{W[n]}(w)\mathrm{d}w,
	\end{equation}
	which implies
	\begin{equation}\label{eq:stay-out-of-middle-conv}
	\int_\varepsilon^{1-\varepsilon} f_{W[n]}(w)\mathrm{d}w \leq \frac{\gamma}{\varepsilon^2(1-\varepsilon)^2} = \mu.
	\end{equation}
	Since \eqref{eq:stay-out-of-middle-conv} holds for all \(\mu > 0\) and \(\varepsilon > 0\), we have
	\begin{equation}
	\lim_{n\to\infty}\mathbb{P}(\varepsilon < W[n] < 1 -\varepsilon) = 0,
	\end{equation}
	or equivalently,	
	\begin{equation}\label{eq:ordinary_conv2}
	\lim_{n\to\infty}\mathbb{P}(W[n]\leq\varepsilon \ \cup \ W[n]\geq 1-\varepsilon) = 1.
	\end{equation}
\end{proof}
	\begin{proof}[Proof of Lemma \ref{lem:conv-prob}]
		We know from Lemma \ref{lem:ms-implies-prob-conv} that for all \(\varepsilon > 0\) and \(\delta>0\) there exists \(N_1>0\) such that for all \(n\geq N_1\),
		\begin{equation}\label{eq:yi-prob-zero-or-one}
		\mathbb{P}(\varepsilon < X_i[n+1] < 1-\varepsilon) < \dfrac{\delta}{2}.
		\end{equation}
		The assumption that \(X_j[n] \xrightarrow{P}0\) as \(n\to\infty\), together with the uniform integrability of \(X_j[n]\) (it is bounded by the interval \([0,1]\)) implies that the expected value of \(X_j[n]\) also converges to \(0\). (This is a standard result in probability theory. See, e.g., \cite[Theorem 5.5.2]{Durrett}.) Thus, for all \(\delta>0\) there exists \(N_2>0\) such that for all \(n\geq N_2\),
		\begin{equation}\label{eq:yj-conv-mean}
		\E[X_j[n]] < \dfrac{\delta}{2}.
		\end{equation}
		We want to show that \(X_i[n]\xrightarrow{P}0\) as \(n\to\infty\). To this end, recall that \(\alpha\in(0,1)\) and that \(t_{ij}>0\) since we assume that \(i\) puts a trust in \(j\). Let \(0<\varepsilon<\alpha t_{ij}\) and \(\delta > 0\). Then for all \(n>\max\{N_1,N_2\}\), we have
		\begin{equation*}
		\begin{aligned}
		&\mathbb{P}(X_i[n+1] > \varepsilon)\\
		&= \mathbb{P}(X_i[n+1] \geq 1-\varepsilon) + \mathbb{P}(\varepsilon < X_i[n+1] < 1-\varepsilon) \\
		&< \mathbb{P}(X_i[n+1] \geq 1-\varepsilon) + \dfrac{\delta}{2}
		\end{aligned}
		\end{equation*}
		\begin{equation}
		\begin{aligned}
		&= \int_{\matr{x}}\mathbb{P}(X_i[n+1]\geq 1 - \varepsilon\mid \matr{X}[n] = \matr{x})f_{\matr{X}[n]}(\matr{x})d\matr{x} + \dfrac{\delta}{2}\\
		&=\int_{\matr{x}}\mathbb{P}\big((1-\alpha)x_i + \alpha\sum_{k=1}^{K}t_{ik}A_k[n]\geq 1 - \varepsilon\mid \matr{X}[n] = \matr{x}\big)\\
		&\qquad\cdot f_{\matr{X}[n]}(\matr{x})d\matr{x} + \dfrac{\delta}{2}\\
		&\leq \int_{\matr{x}}\mathbb{P}(1-\alpha + \alpha (1-t_{ij}(1-A_j[n]))\geq 1 - \varepsilon\mid \matr{X}[n] = \matr{x})\\
		&\qquad \cdot f_{\matr{X}[n]}(\matr{x})d\matr{x} + \dfrac{\delta}{2}\\
		&= \int_{\matr{x}}\mathbb{P}(\alpha t_{ij}(1-A_j[n])\leq \varepsilon\mid \matr{X}[n] = \matr{x}) f_{\matr{X}[n]}(\matr{x})d\matr{x} + \dfrac{\delta}{2}\\
		&= \int_{\matr{x}}\mathbb{P}(1-A_j[n]\leq\dfrac{\varepsilon}{\alpha t_{ij}}\mid \matr{X}[n] = \matr{x}) f_{\matr{X}[n]}(\matr{x})d\matr{x} + \dfrac{\delta}{2}\\
		&\leq \int_{\matr{x}}x_j f_{\matr{X}[n]}(\matr{x})d\matr{x} + \dfrac{\delta}{2}\\
		&= \E[X_j[n]] + \dfrac{\delta}{2} < \dfrac{\delta}{2} + \dfrac{\delta}{2} = \delta,
		\end{aligned}
		\end{equation}
		where the first inequality follows from \eqref{eq:yi-prob-zero-or-one}, the second inequality follows from the facts that \(\sum_{k=1}^{K}t_{ik}A_k[n] \leq \sum_{k=1, k\neq j}^{K}t_{ik} + t_{ij}A_j[n] = 1 - t_{ij} +t_{ij}A_j[n]\) and \(X_i[n]\leq 1\), and the last inequality follows from \eqref{eq:yj-conv-mean}. We have also used the fact that \(A_j[n]\sim\mathrm{Bernoulli}(X_j[n])\) conditioned on \(X_j[n]\).
	\end{proof}
	\bibliographystyle{IEEEtran}
	\bibliography{IEEEabrv,asilomar2019} 

\begin{thebibliography}{10}
\providecommand{\url}[1]{#1}
\csname url@samestyle\endcsname
\providecommand{\newblock}{\relax}
\providecommand{\bibinfo}[2]{#2}
\providecommand{\BIBentrySTDinterwordspacing}{\spaceskip=0pt\relax}
\providecommand{\BIBentryALTinterwordstretchfactor}{4}
\providecommand{\BIBentryALTinterwordspacing}{\spaceskip=\fontdimen2\font plus
\BIBentryALTinterwordstretchfactor\fontdimen3\font minus
  \fontdimen4\font\relax}
\providecommand{\BIBforeignlanguage}[2]{{%
\expandafter\ifx\csname l@#1\endcsname\relax
\typeout{** WARNING: IEEEtran.bst: No hyphenation pattern has been}%
\typeout{** loaded for the language `#1'. Using the pattern for}%
\typeout{** the default language instead.}%
\else
\language=\csname l@#1\endcsname
\fi
#2}}
\providecommand{\BIBdecl}{\relax}
\BIBdecl

\bibitem{proskurnikov2017}
A.~V. Proskurnikov and R.~Tempo, ``A tutorial on modeling and analysis of
  dynamic social networks: Part {I},'' \emph{Annu. Rev. Control}, vol.~43, pp.
  65--79, 2017.

\bibitem{proskurnikov2018}
------, ``A tutorial on modeling and analysis of dynamic social networks: Part
  {II},'' \emph{Annu. Rev. Control}, vol.~45, pp. 166--190, 2018.

\bibitem{degroot}
M.~H. DeGroot, ``Reaching a consensus,'' \emph{J. Amer. Statist. Assoc.},
  vol.~69, no. 345, pp. 118--121, 1974.

\bibitem{mobilia2003}
M.~Mobilia, ``Does a single zealot affect an infinite group of voters?''
  \emph{Phys. Rev. Lett.}, vol.~91, no.~2, 028701, 2003.

\bibitem{sudbury73}
A.~Sudbury and P.~Clifford, ``A model for spatial conflict,''
  \emph{Biometrika}, vol.~60, no.~3, pp. 581--588, Dec. 1973.

\bibitem{yildiz}
E.~Yildiz, A.~Ozdaglar, D.~Acemoglu, A.~Saberi, and A.~Scaglione, ``Binary
  opinion dynamics with stubborn agents,'' \emph{{ACM} Trans. Econ. Comput.},
  vol.~1, no.~4, Dec. 2013.

\bibitem{ghaderi}
J.~Ghaderi and R.~Srikant, ``Opinion dynamics in social networks with stubborn
  agents,'' \emph{Automatica}, vol.~50, no.~12, pp. 3209--3215, Dec. 2014.

\bibitem{wai}
H.~{Wai}, A.~{Scaglione}, and A.~{Leshem}, ``Active sensing of social
  networks,'' \emph{{IEEE} Trans. Signal Inf. Process. Netw.}, vol.~2, no.~3,
  pp. 406--419, Sep. 2016.

\bibitem{acemoglu2010}
D.~Acemoglu, A.~Ozdaglar, and A.~ParandehGheibi, ``Spread of (mis)information
  in social networks,'' \emph{Games Econ. Behav.}, vol.~70, no.~2, pp.
  194--227, 2010.

\bibitem{mukhopadhyay2016}
A.~{Mukhopadhyay}, R.~R. {Mazumdar}, and R.~{Roy}, ``Binary opinion dynamics
  with biased agents and agents with different degrees of stubbornness,'' in
  \emph{Proc. of 2016 28th International Teletraffic Congress (ITC 28)},
  vol.~01, Sep. 2016, pp. 261--269.

\bibitem{hunter2018}
\BIBentryALTinterwordspacing
D.~S. {Hunter} and T.~{Zaman}, ``Opinion dynamics with stubborn agents,''
  \emph{arXiv e-prints}, Jun. 2018. [Online]. Available:
  \url{https://arxiv.org/abs/1806.11253}
\BIBentrySTDinterwordspacing

\bibitem{scaglione}
A.~{Leshem} and A.~{Scaglione}, ``The impact of random actions on opinion
  dynamics,'' \emph{{IEEE} Trans. Signal Inf. Process. Netw.}, vol.~4, no.~3,
  pp. 576--584, Sep. 2018.

\bibitem{HorJoh}
R.~A. Horn and C.~R. Johnson, \emph{Matrix Analysis}, 2nd~ed.\hskip 1em plus
  0.5em minus 0.4em\relax New York, NY, USA: Cambridge University Press, 2013.

\bibitem{Durrett}
R.~Durrett, \emph{Probability: Theory and Examples}, 2nd~ed.\hskip 1em plus
  0.5em minus 0.4em\relax Cambridge, United Kingdom: Cambridge University
  Press, 2010.

\end{thebibliography}
\end{document}